\documentclass{article}

\title{Distributed Detection of Cliques in Dynamic Networks}
\date{}
\author{
	Matthias Bonne
	\thanks{matt.b@cs.technion.ac.il} \\
	Department of Computer Science, Technion, Israel
\and
	Keren Censor-Hillel
	\thanks{ckeren@cs.technion.ac.il} \\
	Department of Computer Science, Technion, Israel
}

\usepackage[left=1in,right=1in]{geometry}
\usepackage{hyperref}
\usepackage{amsmath}
\usepackage{amsthm}
\usepackage{cite}
\usepackage{tikz}

\newtheorem{theorem}{Theorem}[subsection]
\newtheorem{corollary}[theorem]{Corollary}
\newtheorem{lemma}{Lemma}
\newtheorem{openq}{Open Question}
\newtheorem{observation}{Observation}

\newcommand{\TXTMLIST}{\textsf{MemList}}
\newcommand{\TXTMDTCT}{\textsf{MemDetect}}
\newcommand{\TXTLIST}{\textsf{List}}
\newcommand{\TXTDTCT}{\textsf{Detect}}

\newcommand{\MLIST}[1]{$\textsf{MemList}(#1)$}
\newcommand{\MDTCT}[1]{$\textsf{MemDetect}(#1)$}
\newcommand{\LIST}[1]{$\textsf{List}(#1)$}
\newcommand{\DTCT}[1]{$\textsf{Detect}(#1)$}

\newcommand{\pa}[1]{\left(#1\right)}

\newcommand{\pabr}[1]{\left\{#1\right\}}
\newcommand{\ceil}[1]{\left\lceil#1\right\rceil}
\newcommand{\floor}[1]{\left\lfloor#1\right\rfloor}
\newcommand{\A}[1]{\begin{align*}#1\end{align*}}
\newcommand{\AL}[2]{\begin{align}\label{#1}#2\end{align}}

\begin{document}

\maketitle

\begin{abstract}

This paper provides an in-depth study of the fundamental problems of
finding small subgraphs in distributed dynamic networks.

While some problems are trivially easy to handle, such as detecting a
triangle that emerges after an edge insertion, we show that, perhaps
somewhat surprisingly, other problems exhibit a wide range of
complexities in terms of the trade-offs between their round and
bandwidth complexities.

In the case of triangles, which are only affected by the topology of the
immediate neighborhood, some end results are:
\begin{itemize}
	\item
		The bandwidth complexity of $1$-round dynamic triangle
		detection or listing is $\Theta(1)$.
	\item
		The bandwidth complexity of $1$-round dynamic triangle
		membership listing is $\Theta(1)$ for node/edge
		deletions, $\Theta(n^{1/2})$ for edge insertions, and
		$\Theta(n)$ for node insertions.
	\item
		The bandwidth complexity of $1$-round dynamic triangle
		membership detection is $\Theta(1)$ for node/edge
		deletions, $O(\log n)$ for edge insertions, and
		$\Theta(n)$ for node insertions.
\end{itemize}
Most of our upper and lower bounds are \emph{tight}. Additionally, we
provide almost always tight upper and lower bounds for larger cliques.

\end{abstract}

\section{Introduction}

A fundamental problem in many computational settings is to find small
subgraphs. In distributed networks it is particularly vital for various
reasons, among which is the ability to perform some tasks much faster
if, say, triangles do not occur in the underlying network graph (see,
eg.,~\cite{Pettie2015,Hirvonen2017}).

Finding cliques is a \emph{local} task that trivially requires only a
single communication round if the message size is unbounded. However,
its complexity may dramatically increase when the bandwidth is
restricted to the standard $O(\log n)$ bits, for an $n$-node network.
For example, the complexity of detecting 4-cliques is at least
$\Omega(n^{1/2})$~\cite{Czumaj2018}. For triangles, the complexity is
yet a fascinating riddle, where only recently the first non-trivial
complexities of $\tilde{O}(n^{2/3})$ and $\tilde{O}(n^{3/4})$ have been
given for detection and listing, respectively~\cite{Izumi2017}, and the
current state-of-the-art is the $\tilde{O}(n^{1/2})$-round algorithm
of~\cite{Chang2019}. For listing, there is an $\tilde{\Omega}(n^{1/3})$
lower bound~\cite{Izumi2017,Pandurangan0S18}. The only non-trivial lower
bounds for detection say that a single round is insufficient, as given
in~\cite{Abboud2017} and extended for randomized algorithms
in~\cite{Fischer2018}. In~\cite{Abboud2017} it was also shown that
$1$-bit algorithms require $\Omega(\log^* n)$ rounds, improved
in~\cite{Fischer2018} to $\Omega(\log{n})$.

In this paper, we address the question of detecting small cliques in
\emph{dynamic} networks of limited bandwidth. We consider a model that
captures real-world behavior in networks that undergo changes, such as
nodes joining or leaving the network, or communication links between
nodes that appear or disappear. Various problems have been studied in
many variants of such a setting (see Section~\ref{subsec:related}).

The task of finding cliques is unique, in that it is trivial if the
bandwidth is not restricted, and it can be easily guaranteed that at the
end of each round all nodes have the correct output. This implies that
one does not have to wait for \emph{stabilization} and does not have to
assume that the network is \emph{quiet} for any positive number of
rounds. If, however, the bandwidth is restricted, the solution may not
be as simple, although some problems can still be solved even with very
small bandwidth.

As a toy example, consider the case of triangle listing when a new edge
is inserted to the graph. The endpoints of the inserted edge simply
broadcast a single bit that indicates this change to all of their
neighbors, and hence if a new triangle is created then its third
endpoint detects this by receiving two such indications.

Nevertheless, we show that this simplified case is far from reflecting
the general complexities of clique problems in such a dynamic setting.
For example, the above algorithm does not solve the problem of
\emph{membership-listing} of triangles, in which \emph{each} node should
list all triangles that contain it. Indeed, we prove that this stronger
variant \emph{cannot} be solved with constant bandwidth, and, in fact,
every solution must use at least $\Omega(\sqrt{n})$ bits.

Our contributions provide an in-depth study of various detection and
listing problems of small cliques in this dynamic setting, as we
formally define and summarize next.

\subsection{Our contributions}

For a subgraph $H$ we categorize four types of tasks: Detecting an
appearance of $H$ in the network, for which it is sufficient that a
single node does so, listing all appearances of $H$, such that every
appearance is listed at least by a single node, and their two
\emph{membership} variants, membership-detection and membership-listing,
for which each node has to detect whether it is a member of a copy of
$H$, or list all copies of $H$ to which it belongs, respectively.

The model is explicitly defined in Section~\ref{section:prelim}. In a
nutshell, there can be one topology change in every round, followed by a
standard communication round among neighboring nodes, of $B$ bits per
message, where $B$ is the bandwidth. An algorithm takes $r$ rounds if
the outputs of the nodes are consistent with a correct solution after
$r$ communication rounds that follow the last change. In particular, a
$1$-round algorithm is an algorithm in which the outputs are correct
already at the end of the round in which the topology change occurred.
Hence, $1$-round algorithms are very strong, in the sense that they work
even in settings in which no round is free of changes. We note that in
what follows, all of our algorithms are deterministic and all of our
lower bounds hold also for randomized algorithms.

~\\\textbf{Triangles:}
Our upper and lower bounds for triangles ($H=K_3$) are displayed in
Table~\ref{tbl:summary}.

\begin{table}[tbh]
\centering
\begin{tabular}{|l|c|c|c|c|}
\hline

	& Node deletions
	& Edge deletions
	& Edge insertions
	& Node insertions
	\\ \hline

Detection/Listing
	& $0$
	& $\Theta(1)$
	& $\Theta(1)$
	& $\Theta(1)$
	\\ \hline

Membership detection
	& $0$
	& $\Theta(1)$
	& $O(\log n)$
	& $\Theta(n)$
	\\ \hline

Membership listing
	& $0$
	& $\Theta(1)$
	& $\Theta(\sqrt{n})$
	& $\Theta(n)$
	\\ \hline

\end{tabular}
\caption{
	The bandwidth complexities of $1$-round algorithms for dynamic
	triangle problems.
}
\label{tbl:summary}
\end{table}
Most of the complexities in this table are shown to be tight, by
designing algorithms and proving matching lower bounds. The one
exception is for membership detection with edge insertions, for which we
show an algorithm that uses $O(\log n)$ bits of bandwidth, but we do not
know whether this is tight. However, we also show a $1$-round algorithm
for this problem that works with a bandwidth of
$O((\Delta\log{n})^{1/2})$, where $\Delta$ is the maximum degree in the
graph, implying that if our logarithmic algorithm is optimal, then a
proof for its optimality must exhibit a worst case in which the maximum
degree is $\Omega(\log{n})$.

A single round is sufficient for solving all clique problems, given
enough (linear) bandwidth. Nevertheless, for the sake of comparison, we
show that with just one additional communication round, all of the
bandwidth complexities in Table~\ref{tbl:summary} drop to $\Theta(1)$,
apart from membership listing of triangles, whose bandwidth complexity
for $r$-round algorithms is $\Theta(n/r)$.

~\\\textbf{Larger Cliques:}
We also study the bandwidth complexities for finding cliques on $s>3$
nodes. Here, too, a single round is sufficient for all problems, and the
goal is to find the bandwidth complexity of each problem. Some of the
algorithms and lower bounds that we show for triangles carry over to
larger cliques, but others do not. Yet, with additional techniques, we
prove that for cliques of constant size, almost all of the $1$-round
bandwidth complexities are the same as their triangle counterparts.

~\\\textbf{Combining types of changes:}
In most cases, one can obtain an algorithm that handles several types of
changes, by a simple combination of the corresponding algorithms.
However, intriguingly, sometimes this is not the case. A prime example
is when trying to combine edge insertions and node insertions for
triangle detection: a node obtaining 1-bit indications of a change from
two neighbors cannot tell whether this is due to an edge inserted
between them, or due to an insertion of a node that is a neighbor of
both. In some of these cases we provide techniques to overcome these
difficulties, and use them to adjust our algorithms to cope with more
than one type of change.

\subsection{Challenges and techniques}

~\\\textbf{Algorithms:}
The main challenge for designing algorithms is how to convey enough
information about the topology changes that occur, despite non-trivial
(in particular, sublinear) bandwidth. Consider, e.g., edge insertions.
As described, while listing triangles is trivial with a single
indication bit, this fails for membership detection or membership
listing.

For membership detection we can still provide a very simple algorithm
for which a logarithmic bandwidth suffices, by sending the identity of
the new edge, and by \emph{helping} neighbors in the triangle to know
that they are such. For membership listing even this is insufficient. To
overcome this challenge, we introduce a technique for sending and
collecting \emph{digests} of neighborhood information. When all digests
from a given neighbor have been collected, one can determine the entire
neighborhood of this neighbor. The caveat in using such an approach in a
straightforward manner is that a node needs to list its triangles with a
newly connected neighbor already at the same round in which they
connected, and cannot wait to receive all of its neighbor's digests. By
our \emph{specific choice} of a digests, we ensure that a
newly-connected neighbor has enough information to give a correct output
after a single communication round, despite receiving only a single
digest from its latest neighbor.

~\\\textbf{Lower bounds:}
For the lower bounds, our goal is to argue that in order to guarantee
that all nodes give a correct output, each node must receive enough
information about the rest of the graph. To do this, we identify
sequences of topology changes that exhibit a worst-case behavior, in the
sense that a node cannot give a correct output if it receives too little
information.

One approach for doing this is to look at a particular node $x$, and
define as many sequences as possible such that the correct output of $x$
is different for each sequence. At the same time, we ensure that the
number of messages that $x$ receives from the other nodes is as small as
possible. These two requirements are conflicting --- if $x$ can have
many different outputs, it must have many different neighbors, which
implies that it receives many messages with information. Still, we are
able to find such a family of sequences for each problem, and we wrap-up
our constructions by using counting arguments to prove the desired lower
bounds. In some cases, e.g., membership-listing under edge insertions,
even this is insufficient, and we construct the sequences such that one
\emph{critical} edge affects the output of $x$, but it is \emph{added
last}, so that it conveys as little information as possible.

The above can be seen as another step in the spirit of the \emph{fooling
views} framework, which was introduced in~\cite{Abboud2017} for
obtaining the first lower bounds for triangle detection under limited
bandwidth. After being beautifully extended by~\cite{Fischer2018}, our
paper essentially gives another indication of the power of the fooling
views framework in proving lower bounds for bandwidth-restricted
settings.

On the way to constructing our worst-case graph sequences, we prove
combinatorial lemmas that show the existence of graphs with certain
desired properties. To make our lower bounds apply also for randomized
algorithms, we rely on Yao's lemma and on additional machinery that we
develop.

\subsection{Additional related work}
\label{subsec:related}

\textbf{Dynamic distributed algorithms:}
Dynamic networks have been the subject of much research. A central
paradigm is that of \emph{self-stabilization}~\cite{Dolev2000}, in which
the system undergoes various changes, and after some quiet time needs to
converge to a stable state. The model addressed in this paper can be
considered as such, but our focus is on algorithms that do not require
any quiet rounds (though we also address the gain in bandwidth
complexity if the system does allow them, for the sake of comparison).
Yet, we assume a single topology change in each round. A single change
in each round and enough time to recover is assumed in recent algorithms
for maintaining a maximal independent set~\cite{CensorHillel2016,
Assadi2019,GuptaK18,AssadiOSS18,DuZ2018} and matchings~\cite{Solomon16},
and for analyzing amortized complexity~\cite{ParterPS16}. Highly-dynamic
networks, in which multiple topology changes can occur in each round are
analyzed in~\cite{Bamberger2018,CensorHillelDKPS19}, and it is an
intriguing open question how efficient can subgraph detection be in such
a setting. Various other dynamic settings in which the topology changes
can be almost arbitrary have been proposed, but the questions that arise
in such networks address the flow of information and agreement problems,
as they are highly non-structured. A significant paper that relates to
our work is~\cite{KonigW13}, which differs from our setting in the
bandwidth assumption. Clique detection in the latter model is trivial,
and indeed their paper addresses other problems.

\textbf{Distributed subgraph detection:}
In the CONGEST model and in related models, where the nodes
synchronously communicate with $O(\log{n})$ bits, much research is
devoted to subgraph detection, e.g.,~\cite{Fraigniaud2017,Gonen2017,
Korhonen2017,Drucker2014,Abboud2017,Izumi2017,Chang2019,
Pandurangan0S18}.

A related problem is property-testing of subgraphs, where the nodes need
to determine, for a given subgraph $H$, whether the graph is $H$-free or
\emph{far} from being $H$-free~\cite{CensorHillel2019,EvenFischer2017,
EvenLevi2017,Fischer2017,Fraigniaud2016,Brakerski2011}.

\subsection{Preliminaries}
\label{section:prelim}

In the dynamic setting, the network is a sequence of graphs. The initial
graph, whose topology is known to the nodes, represents the state of the
network at some starting point in time, and every other graph in the
sequence is either identical to its preceding graph or obtained from it
by a single topology change. The network is synchronized, and in each
round, each node can send to each one of its neighbors a message of $B$
bits, where $B$ is the bandwidth of the network. Each node has a unique
ID and knows the IDs of all its neighbors. We denote by $n$ the number
of possible IDs. Hence, $n$ is an upper bound on the number of nodes
that can participate in the network at all times.

Note that the nodes do not receive any indication when a topology change
occurs.  A node can deduce that a change has occurred by comparing the
list of its neighbors in the current round and in the previous round.
This implies that a node $u$ cannot distinguish between the insertion of
an edge $uv$ and the insertion of a node $v$ which is connected to $u$,
as the list of neighbors of $u$ is the same in both cases.

An algorithm can be designed to handle edge insertions or deletions or
node insertions or deletions, or any combination of these. We say that
an algorithm is an $r$-round algorithm if the outputs of the nodes after
$r$ rounds of communication starting from the last topology change are
correct. The (deterministic or randomized) $r$-round bandwidth
complexity of a problem is the minimum bandwidth for which an $r$-round
(deterministic or randomized) algorithm exists. We denote by \MLIST{H},
\MDTCT{H}, \LIST{H}, and \DTCT{H}, respectively, the problems of
membership listing, membership detection, listing and detection of $H$.
The following is a simple observation.

\begin{observation}\label{obs:probrel}

Let $r$ be an integer and let $H$ be a graph. Denote by $B_{\TXTMLIST}$,
$B_{\TXTMDTCT}$, $B_{\TXTLIST}$ and $B_{\TXTDTCT}$ the (deterministic or
randomized) $r$-round bandwidth complexities of \MLIST{H}, \MDTCT{H},
\LIST{H}, and \DTCT{H}, respectively, under some type of topology
changes. Then: $B_{\TXTDTCT} \leq B_{\TXTLIST} \leq B_{\TXTMLIST}$ and
$B_{\TXTDTCT} \leq B_{\TXTMDTCT} \leq B_{\TXTMLIST}$.

\end{observation}

Due to this observation, our exposition starts with the more challenging
task of membership-listing, then addresses membership-detection, and
concludes with listing and detection. Within each case, we first provide
algorithms for each type of topology change, and then prove lower
bounds.

\section{Triangle problems}

\subsection{Membership listing}

Table~\ref{tbl:k3_mlist} summarizes our results for membership listing
of triangles.
\begin{table}[tbh]
\centering
\begin{tabular}{|l|c|c|c|c|}
\hline

	& Node deletions
	& Edge deletions
	& Edge insertions
	& Node insertions
	\\ \hline
$r = 1$
	& $0$
	& $\Theta(1)$
	& $\Theta(\sqrt{n})$
	& $\Theta(n)$
	\\ \hline
$r \geq 2$
	& $0$
	& $\Theta(1)$
	& $\Theta(1)$
	& $\Theta(n/r)$
	\\ \hline

\end{tabular}
\caption{
	Bandwidth complexities of \MLIST{K_3}.
}
\label{tbl:k3_mlist}
\end{table}
\subsubsection{Upper bounds}
\label{subsubsec:tri-mlist-upper}

We start by showing that if only node deletions are allowed, no
communication is required.

\begin{theorem}\label{th:k3_mlist_ub_mv}
	The deterministic $1$-round bandwidth complexity of \MLIST{K_3}
	under node deletions is $0$.
\end{theorem}
\begin{proof}

The algorithm is very simple: the output of a node $u$ is the set of all
triangles $uvw$ that exist in the initial graph, such that $v$ and $w$
are both neighbors of $u$ in the final graph.

Since only node deletions are allowed, every triangle $uvw$ that exists
in the final graph must also exist in the initial graph, and all three
nodes $u$, $v$ and $w$ must exist and be connected to each other.
Conversely, if a triangle $uvw$ exists in the initial graph and $u$ is
still connected to $v$ and $w$, then the triangle $uvw$ must exist in
the final graph. Therefore, the output of $u$ is exactly the set of all
triangles that contain $u$ in the final graph.
\qedhere

\end{proof}

Next, we show how to handle edge deletions with $O(1)$ bits of
bandwidth.

\begin{theorem}\label{th:k3_mlist_ub_me}
	The deterministic $1$-round bandwidth complexity of \MLIST{K_3}
	under edge deletions is $O(1)$.
\end{theorem}
\begin{proof}

The idea of the algorithm is simple: every node $u$ needs to maintain
the list of all triangles that contain it. In order to do this, $u$
needs to be able to determine, for every two of its neighbors $v$ and
$w$, whether or not the edge $vw$ exists.

This can be achieved with only a single bit, as follows. Every node
sends to all its neighbors an indication whether or not it lost an edge
on the current round. We denote this bit by \texttt{DELETED}. Thus, for
every triangle $uvw$, whenever one of the edges of the triangle is
deleted (say, $vw$), the other node ($u$ in this case) receives
$\texttt{DELETED} = 1$ from both $v$ and $w$. Since the only possible
change in the graph is the deletion of an edge, $u$ can deduce that the
edge $vw$ was deleted.
\qedhere

\end{proof}

The above algorithm does not work as-is if both edge deletions and node
deletions are allowed. For example, assume there exists a triangle
$uvw$. When we remove the edge $vw$, both $v$ and $w$ send
$\texttt{DELETED} = 1$ to $u$, thus $u$ knows that the edge $vw$ was
deleted and the triangle $uvw$ no longer exists. However, $u$ receives
the exact same input if some other node $x$, which is connected to both
$v$ and $w$, but not to $u$, is deleted. Here, too, both $v$ and $w$
lose a neighbor, thus they both send $\texttt{DELETED} = 1$ to $u$, and
$u$ deduces --- now incorrectly --- that the triangle $uvw$ no longer
exists and removes it from its output.

This problem can be fixed, by observing that when the edge $vw$ is
deleted, both $v$ and $w$ know whether or not $x$ is connected to $u$
(i.e., whether or not the triangle $vxu$, or $wxu$, exists). Therefore,
we can have each node $v$ send $\texttt{DELETED} = 1$ to each neighbor
$u$ only if the lost neighbor was connected to $u$; that is, $v$ sends
$\texttt{DELETED} = 1$ to $u$ if it lost an edge to some node $x$ such
that the triangle $vux$ existed on the previous round. This information
allows $u$ to distinguish between the two cases described above and give
the correct output. Hence the following corollary.

\begin{corollary}\label{cor:k3_mlist_ub_me_mv}
	The deterministic $1$-round bandwidth complexity of \MLIST{K_3}
	under node/edge deletions is $O(1)$.
\end{corollary}

Handling edge insertions and node insertions is much more complicated.
We start by showing an algorithm that handles edge insertions.

\begin{theorem}\label{th:k3_mlist_ub_pe}
	The deterministic $1$-round bandwidth complexity of \MLIST{K_3}
	under edge insertions is $O(\sqrt{n})$.
\end{theorem}
\begin{proof}

Let $N_u(r)$ be the set of neighbours of $u$ on round $r$. Note that
$N_u(r)$ can be encoded as an $n$-bit string, which indicates, for every
node $x$, whether or not $x$ is a neighbour of $u$ on round $r$.

The algorithm is as follows. When a new edge $uv$ is inserted on round
$r$, both $u$ and $v$ send to all their neighbors the identity of their
new neighbor, denoted by \texttt{NEWID}.

In addition, $u$ sends a bitmask of $\ceil{\sqrt{n}}$ bits to $v$,
indicating, for every one of the previous $\ceil{\sqrt{n}}$ rounds,
whether or not one of $u$'s neighbors has sent a \texttt{NEWID} to $u$.
We denote this information by $\texttt{LAST}_u$. Node $v$ sends to $u$
the same information.

Finally, $u$ encodes $N_u(r)$ as an $n$-bit string, denoted
$\texttt{ALL}_u(r)$, and starts sending it to $v$ in chunks of
$\ceil{\sqrt{n}}$ bits per round. This process begins on round $r$ and
ends on round $r + \floor{\sqrt{n}}$, when the entire string has been
sent. During these rounds, $u$ keeps sending \texttt{NEWID} to $v$, and
to all its other neighbors, as described above. Node $v$ does the same.
It should be noted that this continuous communication between $u$ and
$v$ is not intended for allowing them to detect triangles that appear by
the insertion of the edge $uv$, as these are detected immediately due to
previous information. Rather, communicating \texttt{ALL} allows $u$ and
$v$ to detect triangles that appear by other topology changes that may
occur in subsequent rounds, as we show below.

Overall, \texttt{NEWID} requires $O(\log n)$ bits, while \texttt{LAST}
and \texttt{ALL} require $O(\sqrt{n})$ bits. Thus, the required
bandwidth is $O(\sqrt{n})$.

We show that at the end of each round, every node $u$ has enough
information to determine, for every two of its neighbors $v$ and $w$,
whether or not the edge $vw$ exists.

First, since only edge insertions are considered, if the edge $vw$
exists in the initial graph, it exists throughout. Also, if $vw$ is
inserted when at least one of the edges $uv$ or $uw$ already exists,
then $u$ receives this information through \texttt{NEWID}. The only
other case is when $vw$ does not exist in the initial graph, and is
inserted when $u$ is not yet connected to either $v$ or $w$. That is,
the initial graph does not contain any of these three edges, and $vw$ is
inserted before the other two. W.l.o.g. assume that $uv$ is inserted
before $uw$. Now, let $t_v$ be the round in which $uv$ is inserted, and
let $t_w$ be the round in which $uw$ is inserted.

If $t_w - t_v \leq \ceil{\sqrt{n}}$, then when $uv$ is inserted, $v$
sends a \texttt{NEWID} to $w$. Therefore, when $uw$ is inserted, $u$ can
determine from $\texttt{LAST}_w$ that in round $t_v$ a neighbor of $w$
has sent a \texttt{NEWID} to $w$. Since the only edge inserted in that
round is $uv$, $u$ determines that $v$ is a neighbor of $w$. If
$t_w - t_v > \ceil{\sqrt{n}}$, then, in round $t_w$, $v$ has already
sent the entire string $\texttt{ALL}_v(t_v)$ to $u$, which indicates
that $w$ is a neighbor of $v$. Therefore, in all cases, $u$ determines
whether or not the edge $vw$ exists, as claimed.
\qedhere

\end{proof}

The algorithm given for Theorem~\ref{th:k3_mlist_ub_pe} can be extended
to handle edge deletions and node deletions as well, by simply having
each node send to all its neighbors an indication whether or not it lost
a neighbor on the current round, along with the ID of the lost neighbor
(if any). This requires only $O(\log n)$ bits of bandwidth, and one can
verify that this extended algorithm allows every node to give the
correct output in all cases, as required. Therefore:

\begin{corollary}\label{cor:k3_mlist_ub_pe_mvme}
	The deterministic $1$-round bandwidth complexity of \MLIST{K_3}
	under node deletions and edge deletions/insertions is
	$O(\sqrt{n})$.
\end{corollary}

Also, if we are promised a quiet round, the problem can be solved with
$O(1)$ bits of bandwidth, as the next theorem shows.

\begin{theorem}\label{th:k3_mlist_ub_r2_pe}
	The deterministic $2$-round bandwidth complexity of \MLIST{K_3}
	under edge insertions is $O(1)$.
\end{theorem}
\begin{proof}

On every round, for every two neighbors $u$ and $v$, $u$ sends two bits
to $v$:
\begin{enumerate}
\item
	An indication whether or not $u$ has a new neighbor $w$ (i.e.,
	the edge $uw$ has been inserted on the current round). We denote
	this information by \texttt{NEW}.
\item
	An indication whether or not any other neighbor of $u$ has sent
	$\texttt{NEW} = 1$ on the previous round. We denote this
	information by \texttt{LAST}.
\end{enumerate}

Consider a triangle $uvw$, and assume that the edge $uv$ is the last one
inserted. When $uv$ is inserted, both $u$ and $v$ send
$\texttt{NEW} = 1$ to $w$. Since only one edge can be inserted on a
single round, $w$ knows that the edge $uv$ has been inserted, and
therefore the triangle $uvw$ exists.

On the next round, since $w$ received $\texttt{NEW} = 1$ from $v$, it
sends $\texttt{LAST} = 1$ to $u$, indicating that some other neighbor of
$w$ has sent $\texttt{NEW} = 1$ on the last round. $u$ knows that the
edge $uv$ is the only edge that has been inserted on the last round,
hence $v$ must be a neighbor of $w$. Therefore, $u$ can deduce that the
triangle $uvw$ exists. Likewise, $w$ sends $\texttt{LAST} = 1$ to $v$,
from which $v$ can deduce that the triangle $uvw$ exists.

It follows that all three nodes can determine that the triangle $uvw$
exists after two rounds of communication, as required.
\qedhere

\end{proof}

The above algorithm can be combined with a variant of the algorithm of
Corollary~\ref{cor:k3_mlist_ub_me_mv}, in order to handle node deletions
and edge deletions as well. In that algorithm, every node $v$ sends a
bit to every neighbor $u$, indicating whether $v$ lost an edge on the
current round. We denoted this information by \texttt{DELETED}. We can
do the same thing here; however, we have the same problem that we
encountered in Corollary~\ref{cor:k3_mlist_ub_me_mv} --- when a node $u$
receives $\texttt{DELETED} = 1$ from two neighbors $v$ and $w$ on the
same round, it cannot determine whether the edge between $v$ and $w$ was
deleted, or some other node $x$, which was connected to both $v$ and
$w$, was deleted. The triangle $uvw$ exists in the latter case, but not
in the former case. Therefore, $u$ must be able to distinguish between
these two cases.

In Corollary~\ref{cor:k3_mlist_ub_me_mv} we solved this problem by
having every node $v$ send $\texttt{DELETED} = 1$ to a neighbor $u$ only
if it lost a neighbor which is currently connected to $u$. This is
possible, because the algorithm of Corollary~\ref{cor:k3_mlist_ub_me_mv}
guarantees that on every round, $v$ knows the list of all common
neighbors of $v$ and $u$.

In our case there is no such guarantee --- $u$ will know whether or not
its missing neighbor was connected to $v$ only in the next round.
However, by using the extra quiet round, $u$ can easily distuinguish
between the two problematic cases described above, by checking whether
it received $\texttt{DELETED} = 1$ from some other node on the previous
round.

Thus, we extend the algorithm of Theorem~\ref{th:k3_mlist_ub_r2_pe} as
follows. Whenever a node $u$ loses an edge, it sends
$\texttt{DELETED} = 1$ to all its neighbors. On the next round, $u$
sends to all its neighbors an indication whether or not it received
$\texttt{DELETED} = 1$ from any other neighbor on the previous round.
One can verify that this information allows every node to give the
correct output after $2$ rounds of communication. Like the original
algorithm, this uses $O(1)$ bits of bandwidth on every round, thus the
total bandwidth is still $O(1)$ bits. This implies the following
corollary, which is clearly optimal.

\begin{corollary}\label{cor:k3_mlist_ub_r2_pe_mvme}
	The deterministic $2$-round bandwidth complexity of \MLIST{K_3}
	under node deletions and edge deletions/insertions is $O(1)$.
\end{corollary}

Node insertions are harder to handle than the other types of changes. If
we are promised $r-1$ quiet rounds, a simple algorithm exists that uses
only $O\pa{\frac{n}{r}}$ bits of bandwidth. As we show in
Section~\ref{subsubsec:tri-mlist-lower}, this is tight.

\begin{theorem}\label{th:k3_mlist_ub_r_any}
	For every $r$, the deterministic $r$-round bandwidth complexity
	of \MLIST{K_3} under any type of change is $O(n/r)$.
\end{theorem}
\begin{proof}

The algorithm is as follows: on every round, every node $u$ prepares an
$n$-bit message that specifies its current list of neighbors. Then, it
breaks this message into $B$-bit blocks, where $B=\ceil{\frac{n}{r}}$,
and sends it to all its neighbors, one block on every round.

Additionally, on every round, $u$ sends to all its neighbors one bit
indicating whether or not its list of neighbors has changed on the
current round. Whenever the list of neighbors changes, $u$ builds its
new list of neighbors, breaks it into blocks, and starts sending it
again.

After at most $r-1$ quiet rounds, all nodes are guaranteed to finish
sending their list to all their neighbors, thus every node can determine
whether any pair of its neighbors are connected to each other or not, as
required.
\qedhere

\end{proof}

\subsubsection{Lower bounds}
\label{subsubsec:tri-mlist-lower}

The 1-round bandwidth complexities for node and edge deletions are
clearly tight. Next, we show that handling edge insertions in $1$ round
requires at least $\Omega(\sqrt{n})$ bits of bandwidth. By
Theorem~\ref{th:k3_mlist_ub_pe}, this bound is tight. We first prove the
following lemma, which shows that a sufficiently dense bipartite graph
includes a large enough complete bipartite subgraph. This has the same
spirit as the results of Erd\"{o}s~\cite{Erdos1964}, used
in~\cite{Abboud2017,Fischer2018} for bounding the number of single-bit
rounds for detecting triangles, but here the sides can be of different
sizes.

\begin{lemma}\label{lm:densebip}
	For every $\varepsilon \in (0,1)$ there exist $\alpha,\beta,
	\gamma \in (0,1)$, such that for every bipartite graph
	$G=(L \cup R,E)$ having at least $(1-\varepsilon) |L| |R|$
	edges, there are subsets $A \subseteq L$ and $B \subseteq R$,
	whose sizes are $|A| \geq \alpha \cdot |L|$ and
	$|B| \geq \beta \cdot \gamma^{|L|} \cdot |R|$, such that
	$uv \in E$ for every $u \in A$ and $v \in B$ (i.e., the induced
	subgraph on $A \cup B$ is a complete bipartite graph).
\end{lemma}
\begin{proof}

Let $\alpha = \frac{1-\varepsilon}{6}$, and let $\mathcal{A}$ be the set
of all subsets of $L$ whose size is exactly $\ceil{\alpha \cdot |L|}$.
For every $A \in \mathcal{A}$, we denote by $N_A$ the set of all
vertices in $R$ which are connected to every vertex in $A$. Consider the
sum $S = \sum_{A \in \mathcal{A}} |N_A|$. Let $M =
\max \{ |N_A| : A \in \mathcal{A} \}$. Then:
\AL{eq:mlist_pe_sum_leq}{
		S
	\leq
		\sum_{A \in \mathcal{A}} M
	=
		\binom{|L|}{\ceil{\alpha \cdot |L|}} \cdot M
}
On the other hand, the sum $S$ can be computed by counting, for every
$v \in R$, the number of sets $A \in \mathcal{A}$ such that $v \in N_A$:
\A{
		S
	&=
		\sum_{A \in \mathcal{A}} |N_A|
	=
		\sum_{A \in \mathcal{A}} \sum_{v \in N_A} 1
	=
		\sum_{v \in R}
		\sum_{\substack{A \in \mathcal{A} : \\ v \in N_A}} 1
	=
		\sum_{v \in R} | \{ A \in \mathcal{A} : v \in N_A \} |
}

For $p \in (0,1)$, let $k_p$ be the number of vertices in $R$ whose
degree is at least $p \cdot |L|$. For every $v \in R$ whose degree is
$d(v)$ we have $| \{ A \in \mathcal{A} : v \in N_A \} | =
\binom{d(v)}{\ceil{\alpha \cdot |L|}}$. Therefore we can bound the above
sum
$$
		S
	=
		\sum_{v \in R}
			\binom
				{d(v)}
				{\ceil{\alpha \cdot |L|}}
	\geq
		\sum_{\substack{v \in R : \\ d(v) \geq p \cdot |L|}}
			\binom
				{d(v)}
				{\ceil{\alpha \cdot |L|}}
	\geq
		\sum_{\substack{v \in R : \\ d(v) \geq p \cdot |L|}}
			\binom
				{\ceil{p \cdot |L|}}
				{\ceil{\alpha \cdot |L|}}
	=
		k_p \cdot
			\binom
				{\ceil{p \cdot |L|}}
				{\ceil{\alpha \cdot |L|}}
.
$$
Combining this with~\eqref{eq:mlist_pe_sum_leq} gives
$
		\binom{|L|}{\ceil{\alpha \cdot |L|}} \cdot M
	\geq
		k_p \cdot
			\binom
				{\ceil{p \cdot |L|}}
				{\ceil{\alpha \cdot |L|}},
$
which implies:
\AL{eq:mlist_pe_Mbound}{
		M
	\geq
			k_p
		\cdot
			\binom
				{\ceil{p \cdot |L|}}
				{\ceil{\alpha \cdot |L|}}
		/
			\binom
				{|L|}
				{\ceil{\alpha \cdot |L|}}
}

We can bound $k_p$ as follows:
\A{
		|E|
	&=
		\sum_{v \in R} d(v)
	=
		\sum_{\substack{v \in R : \\ d(v) \geq p \cdot |L|}}
			d(v)
		+
		\sum_{\substack{v \in R : \\ d(v) < p \cdot |L|}}
			d(v)
	\leq
		\sum_{\substack{v \in R : \\ d(v) \geq p \cdot |L|}}
			|L|
		+
		\sum_{\substack{v \in R : \\ d(v) < p \cdot |L|}}
			p \cdot |L|
	\\ &=
			k_p \cdot |L|
		+
			(|R| - k_p) \cdot p \cdot |L|
	=
			(1-p) \cdot k_p \cdot |L|
		+
			p \cdot |L| \cdot |R|
}
On the other hand we have
$|E| \geq (1-\varepsilon) \cdot |L| \cdot |R|$, and therefore
$
		(1-p) \cdot k_p \cdot |L| + p \cdot |L| \cdot |R|
	\geq
		(1-\varepsilon) \cdot |L| \cdot |R|
$, which implies $k_p \geq \frac{1-\varepsilon-p}{1-p} \cdot |R|$.
Setting $p = \frac{1-\varepsilon}{2}$ gives
$
	k_p \geq \frac{1-\varepsilon}{1+\varepsilon} \cdot |R|
$, and substituting this into~\eqref{eq:mlist_pe_Mbound} gives:
\A{
		M
	\geq
			\frac
			{
				1-\varepsilon
			}
			{
				1+\varepsilon
			}
		\cdot
			\frac
			{
				\binom
					{\ceil{p \cdot |L|}}
					{\ceil{\alpha \cdot |L|}}
			}
			{
				\binom
					{|L|}
					{\ceil{\alpha \cdot |L|}}
			}
		\cdot
			|R|
}
Finally, we bound the binomial fraction on the right-hand size as
follows. For simplicity, define
$
	a = |L|,
	b = \ceil{p \cdot |L|},
	c = \ceil{\alpha \cdot |L|}
$. Note that $a \geq b \geq c$. Now:
\A{
		\frac
		{
			\binom{b}{c}
		}
		{
			\binom{a}{c}
		}
	&=
		\frac
		{
			\prod_{i=1}^c (i+b-c)
		}
		{
			\prod_{i=1}^c (i+a-c)
		}
	=
		\prod_{i=1}^c \frac{i+b-c}{i+a-c}
	=
		\prod_{i=1}^c (1 - \frac{a-b}{i+a-c})
	\\ &\geq
		\prod_{i=1}^c (1 - \frac{a-b}{a-c})
	=
		\prod_{i=1}^c \frac{b-c}{a-c}
	=
		\pa{\frac{b-c}{a-c}}^c
}
Therefore:
\AL{eq:mlist_pe_binfrac}{
		\frac
		{
			\binom
				{\ceil{p \cdot |L|}}
				{\ceil{\alpha \cdot |L|}}
		}
		{
			\binom
				{|L|}
				{\ceil{\alpha \cdot |L|}}
		}
	\geq
		\pa{
			\frac
			{
				\ceil{p \cdot |L|}
				-
				\ceil{\alpha \cdot |L|}
			}
			{
				|L|
				-
				\ceil{\alpha \cdot |L|}
			}
		}^{
			\ceil{\alpha \cdot |L|}
		}
}

In order to choose appropriate values for $\beta$ and $\gamma$, we now
distinguish between two cases for the graph $G$. If $|L| \leq
\frac{1}{\alpha}$, then $\alpha |L| \leq 1$. Since there must be a
vertex in $|L|$ having at least $(1-\varepsilon) |R|$ neighbors in $R$,
the claim holds for every $\beta \leq 1$ and $\gamma \leq
1-\varepsilon$.

If $|L| > \frac{1}{\alpha}$, we can further develop the right hand side
of~\eqref{eq:mlist_pe_binfrac} as follows:
\A{
		\pa{
			\frac
			{
				\ceil{p |L|}-\ceil{\alpha |L|}
			}
			{
				|L|-\ceil{\alpha |L|}
			}
		}^{
			\ceil{\alpha |L|}
		}
	&\geq
		\pa{
			\frac
			{
				p |L|-(\alpha |L|+1)
			}
			{
				|L|-\alpha |L|
			}
		}^{
			\alpha |L|+1
		}
	=
		\pa{
			\frac
				{p-\alpha}
				{1-\alpha}
			-
			\frac
				{1}
				{(1-\alpha) |L|}
		}^{
			\alpha |L|+1
		}
	\\ &>
		\pa{
			\frac
				{p-\alpha}
				{1-\alpha}
			-
			\frac
				{1}
				{(1-\alpha) \frac{1}{\alpha}}
		}^{
			\alpha |L|+1
		}
	=
		\pa{
			\frac
				{p-2\alpha}
				{1-\alpha}
		}^{
			\alpha |L|+1
		}
	=
		\pa{
			\frac
				{1-\varepsilon}
				{5+\varepsilon}
		}^{
			\alpha |L|+1
		}
}

To sum it all up, we now have
$
		M
	>
		\pa{
			\frac
				{1-\varepsilon}
				{5+\varepsilon}
		}^{
			\alpha \cdot |L|+1
		}
$. Recalling the definition of $M$, this inequality implies that there
exists $A \subseteq L$ whose size is exactly $\ceil{\alpha \cdot |L|}$,
such that
$
		|N_A|
	>
		\pa{
			\frac
				{1-\varepsilon}
				{5+\varepsilon}
		}^{
			\alpha \cdot |L|+1
		}
$, i.e., there are more than
$
	\pa{
		\frac
			{1-\varepsilon}
			{5+\varepsilon}
	}^{
		\alpha \cdot |L|+1
	}
$ vertices in $R$ which are connected to every vertex in $A$. Therefore,
the claim holds for every $\beta$ and $\gamma$ such that:
$
		\beta
	\leq
		\frac
			{1-\varepsilon}
			{5+\varepsilon},
		\gamma
	\leq
		\pa{
			\frac
				{1-\varepsilon}
				{5+\varepsilon}
		}^{\alpha}
	=
		\pa{
			\frac
				{1-\varepsilon}
				{5+\varepsilon}
		}^{
			\frac
				{1-\varepsilon}
				{6}
		}
$. Thus, given $\varepsilon \in (0,1)$, the following values of $\alpha,
\beta,\gamma$ satisfy the claim for every $G$:
\A{
	\alpha =
		\frac
			{1-\varepsilon}
			{6},
	\beta =
		\frac
			{1-\varepsilon}
			{5+\varepsilon},
	\gamma =
		\min \pabr{
			1-\varepsilon,
			\pa{
				\frac
					{1-\varepsilon}
					{5+\varepsilon}
			}^{
				\frac
					{1-\varepsilon}
					{6}
			}
		}
	.
\qedhere
}

\end{proof}

\begin{theorem}\label{th:k3_mlist_lb_pe}
	The randomized $1$-round bandwidth complexity of \MLIST{K_3}
	under edge insertions is $\Omega(\sqrt{n})$.
\end{theorem}
\begin{proof}

Let $A$ be a (randomized) $1$-round algorithm that solves \MLIST{K_3}
under edge insertions using bandwidth $B$ with error probability
$\varepsilon$. Let $t$ be a parameter to be defined later, and consider
a tripartite graph with $n$ nodes as in Figure~\ref{fig:tri}.
\begin{figure}[t]
	\begin{center}
		\begin{tikzpicture}[scale=0.15]
		\tikzstyle{every node}+=[inner sep=0pt]

		\draw [black] (20,-30) ellipse (6 and 14);
		\draw [black] (40,-30) ellipse (3 and 8);

		\draw (20,-10) node {$W$};
		\draw (20,-20) node {$w_1$};
		\draw (20,-25) node {$\vdots$};
		\draw (20,-30) node {$\vdots$};
		\draw (20,-35) node {$\vdots$};
		\draw (20,-40) node {$w_{n-t-1}$};

		\draw (40,-10) node {$U$};
		\draw (40,-25) node {$u_1$};
		\draw (40,-30) node {$\vdots$};
		\draw (40,-35) node {$u_t$};

		\draw (60,-30) node {$v$};

		\draw [black] (59,-30) -- (42.3,-25);
		\draw [black] (59,-30) -- (43,-30);
		\draw [black] (59,-30) -- (42.3,-35);
		\draw [black,dashed,line width=2.8,dash pattern=on 6pt off 6pt] (37,-30) -- (26,-30);
		\draw [black] (60,-31) .. controls (60,-54) and (0,-54) .. (14,-30);

		\end{tikzpicture}
		\caption{
			The lower bound sequence for \MLIST{K_3} with
			edge insertions. The connections between $W$ and
			$U$ are chosen. Then, $v$ is connected to all of
			$U$. The single edge from $v$ to $W$ is added
			last.
		}
		\label{fig:tri}
	\end{center}
\end{figure}
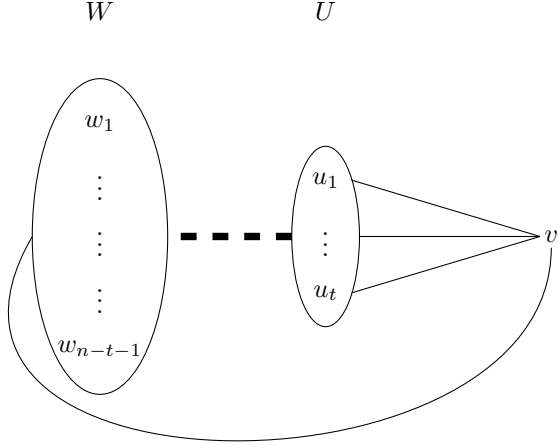
Let $\mathcal{C}$ be the set of all possible bipartite graphs with
vertex sets $W$ and $U$. Note that $|\mathcal{C}| = 2^{t(n-t-1)}$. For
every $C \in \mathcal{C}$ and every $w \in W$, we define a sequence of
changes $S_{C,w}$ as follows. We start with no edges. Then, we insert
edges between $U$ and $W$ to get the bipartite graph $C$. During the
next $t$ rounds, we connect $v$ to every $u \in U$, one by one. Finally,
we insert the edge $vw$. In the end of $S_{C,w}$, after $1$ additional
round of communication, node $v$ must output a list of all triangles
containing $v$. By construction, this list is uniquely determined by the
set of neighbors of $w$ in $U$. Thus, node $v$ needs to know the set of
all neighbors of $w$ after $1$ additional round of communication.

We assume that the output of $v$ is correct with probability at least
$1-\varepsilon$. Therefore, by Yao's lemma, there exists a deterministic
algorithm $A'$ that solves the same problem correctly for at least
$1-\varepsilon$ of all inputs. We define a bipartite graph with node
sets $\mathcal{C}$ and $W$, such that the edge $w-C$ exists if and only
if the output of $v$ is correct for the sequence $S_{C,w}$.

We have $|W| = n-t-1$ and $|\mathcal{C}| = 2^{t(n-t-1)}$. Also, by
assumption, this graph contains at least $(1-\varepsilon)$ of all
possible edges. Therefore, by Lemma~\ref{lm:densebip}, there exists $W_1
\subseteq W$ of size at least $\alpha \cdot (n-t-1)$ and $\mathcal{C}_1
\subseteq \mathcal{C}$ of size at least $\beta \cdot \gamma^{n-t-1}
\cdot 2^{t(n-t-1)}$, for some $\alpha,\beta,\gamma \in (0,1)$ that
depend only on $\varepsilon$, such that the output of $v$ is correct for
the sequence $S_{C,w}$ for every $C \in \mathcal{C}_1$ and every $w \in
W_1$.

Now, consider the input of node $v$ during any sequence $S_{C,w}$.
During the first stage of building the bipartite graph $C$, $v$ is
isolated and receives no input. Then, it receives a set of messages from
the nodes in $U$, and finally one additional message from $w$. Note that
the messages $v$ receives from the nodes in $U$ depend only on $C$, and
not on $w$, since the nodes in $U$ cannot know the identity of $w$ until
the final round of communication.

Now, on every round, every node in $U$ can send to $v$ any of $2^B$
possible messages. Altogether, during the entire sequence, the nodes of
$U$ send to $v$ a set of $\frac{t(t+3)}{2}$ messages. Hence, the number
of possible inputs from the nodes of $U$ to $v$ is
$
	2^{
		B \cdot \frac{t(t+3)}{2}
	}
$. Therefore, there exists $\mathcal{C}_2 \subseteq \mathcal{C}_1$ of
size at least
$
		2^{
			- B \cdot \frac{t(t+3)}{2}
		}
	\cdot
		|\mathcal{C}_1|
$, such that $v$ receives the same input from all nodes in $U$ for every
sequence $S_{C,w}$ for $C \in \mathcal{C}_2$ and every $w \in W_1$.
Thus, we have:
\begin{align}
\label{eq:mlist_ae_lb_1}
		|\mathcal{C}_2|
	&\geq
		2^{
			- B \frac{t(t+3)}{2}
		}
		|\mathcal{C}_1|
	\geq
		2^{
			- B \frac{t(t+3)}{2}
		}
		\beta
		\gamma^{n-t-1}
		2^{t(n-t-1)}
	=
		\beta
		\gamma^{n-t-1}
		2^{
			t(n-t-1) - B \frac{t(t+3)}{2}
		}
\end{align}

On the other hand, we can bound the size of $\mathcal{C}_2$ by
considering the number of possible neighbors of $w$ in any $C \in
\mathcal{C}_2$, for every $w \in W$. Recall that during the sequence
$S_{C,w}$, $v$ receives only one message from $w$. Also, for every $w
\in W_1$, $v$ must determine the set of all neighbors of $w$ in $U$.
Since there are only $2^B$ possible inputs $v$ can receive from $w$,
every $w \in W_1$ can have at most $2^B$ possible sets of neighbors in
any $C \in \mathcal{C}_2$.

Every $w \in W \setminus W_1$ can have any subset of $U$ as its set of
neighbors, hence it can have at most $2^t$ possible sets of neighbors.
Now, since every $C \in \mathcal{C}_2$ is uniquely determined by set of
neighbors of every $w \in W$, we have:
\begin{align}
\label{eq:mlist_ae_lb_2}
		|\mathcal{C}_2|
	\leq
		\pa{
			\prod_{w \in W_1} 2^B
		}
		\pa{
			\prod_{w \in W \setminus W_1} 2^t
		}
	=
		2^{
			B \cdot |W_1|
		}
		\cdot
		2^{
			t \cdot (|W| - |W_1|)
		}
	=
		2^{
			(B-t) \cdot |W_1|
			+
			t(n-t-1)
		}
\end{align}
Combining~\eqref{eq:mlist_ae_lb_1} and~\eqref{eq:mlist_ae_lb_2} gives:
$
			\beta
		\cdot
			\gamma^{n-t-1}
		\cdot
			2^{
				t(n-t-1) - B \cdot \frac{t(t+3)}{2}
			}
	\leq
		2^{
			(B-t) \cdot |W_1| + t(n-t-1)
		}
$,
and setting $t = \ceil{\sqrt{n}}$ gives, with some algebraic
manipulations, $B \geq \Omega(\sqrt{n})$.
\qedhere

\end{proof}

Finally, for node insertions, we show that every $r$-round algorithm
must use at least $\Omega(n/r)$ bits of bandwidth, which is tight by
Theorem~\ref{th:k3_mlist_ub_r_any}.

\begin{theorem}\label{th:k3_mlist_lb_r_pv}
	For every $r$, the randomized $r$-round bandwidth complexity of
	\MLIST{K_3} under node insertions is $\Omega\pa{\frac{n}{r}}$.
\end{theorem}
\begin{proof}

Let $A$ be a (randomized) $r$-round algorithm that solves \MLIST{K_3}
under node insertions using bandwidth $B$ with error probability
$\varepsilon$. We show that $r \cdot B = \Omega(n)$.

Let $u$ be any node, and let $\mathcal{C}$ be the set of all possible
graphs on the other $n-1$ nodes. For every $C \in \mathcal{C}$ we define
the sequence $S_C$ as follows:
\begin{itemize}
\item
	Start with an empty graph (no nodes and no edges).
\item
	During $n-1$ rounds, insert one node on each round and connect
	it to the nodes which have been already inserted, according to
	the edges in $C$. After $n-1$ rounds, the graph is identical to
	the graph $C$.
\item
	On round $n$ insert $u$ and connect to all the other nodes.
\end{itemize}

After $r-1$ quiet rounds $u$ needs to output the list of all triangles
that contain $u$. Since $u$ is connected to all the other nodes, this
implies that $u$ needs to know the graph $C$. For every $C \in
\mathcal{C}$, the output of $u$ is guaranteed to be correct for the
sequence $S_C$ with probability at least $(1-\varepsilon)$. Therefore,
by Yao's lemma, there exists a deterministic algorithm, $A'$, that
guarantees that the output of $u$ is correct for at least
$(1-\varepsilon)$ of all sequences. That is, there exists a subset
$\mathcal{C}_1 \subseteq \mathcal{C}$, whose size is at least
$(1-\varepsilon) \cdot |\mathcal{C}|$, such that $A'$ guarantees the
correct output of $u$ for sequences $S_C$ for all $C \in \mathcal{C}_1$.

Now, the number of possible graphs on $n-1$ nodes is
$2^{\binom{n-1}{2}}$, hence the size of $\mathcal{C}_1$ is at least
$(1-\varepsilon) \cdot 2^{\binom{n-1}{2}}$. Since $A'$ is deterministic,
and $u$ needs to distinguish correctly between all possible $C \in
\mathcal{C}_1$, this implies that the input $u$ receives from its
neighbors must have at least $(1-\varepsilon) \cdot 2^{\binom{n-1}{2}}$
possible values. Every neighbor of $u$ can send any of $2^B$ messages on
every round, thus the number of possible inputs $u$ can receive on a
single round is $2^{B \cdot (n-1)}$. Therefore, during $r$ rounds of
communication, the number of possible inputs to $u$ is
$2^{r \cdot B \cdot (n-1)}$.

Combining the above we get that
$
		2^{
			r \cdot B \cdot (n-1)
		}
	\geq
		(1-\varepsilon) \cdot 2^{\binom{n-1}{2}}
$,
which can be simplified to
$
		r \cdot B
	\geq
		\frac
			{n}
			{2}
		+
		\frac
			{\log (1-\varepsilon)}
			{n-1}
$,
and it follows that $r \cdot B = \Omega(n)$.
\qedhere

\end{proof}

\subsection{Membership detection}

Table~\ref{tbl:k3_mdtct} summarizes the results of this subsection.

\begin{table}[tbh]
\centering
\begin{tabular}{|l|c|c|c|c|}
\hline

	& Node deletions
	& Edge deletions
	& Edge insertions
	& Node insertions
	\\ \hline
$r = 1$
	& $0$
	& $\Theta(1)$
	& $O(\log n)$
	& $\Theta(n)$
	\\ \hline
$r \geq 2$
	& $0$
	& $\Theta(1)$
	& $\Theta(1)$
	& $\Theta(1)$
	\\ \hline

\end{tabular}
\caption{
	Bandwidth complexities of \MDTCT{K_3}
}
\label{tbl:k3_mdtct}
\end{table}
\subsubsection{Upper bounds}

The upper bounds for node deletions and edge deletions follow from
Observation~\ref{obs:probrel}. The following shows that edge insertions
can be handled with $O(\log n)$ bits of bandwidth.

\begin{theorem}\label{th:k3_mdtct_ub_pe}
	The deterministic $1$-round bandwidth complexity of \MDTCT{K_3}
	under edge insertions is $O(\log n)$.
\end{theorem}
\begin{proof}

The algorithm works as follows. On every round, every node sends to all
its neighbors an indication whether or not it got a new neighbor on the
current round, along with the ID of the new neighbor (if any). We denote
this information by \texttt{NEWID}. Note that with just this
information, for every pair of adjacent edges $u-v-w$, at least one of
$u$ and $w$ know that these two edges exist (specifically, the first one
connected to $v$ knows that the other one is also connected to $v$).

Additionally, every node $u$ sends to every neighbor $v$ one bit,
indicating whether $u$ knows that $v$ is part of some triangle. We
denote this bit by \texttt{ACCEPT}.

Suppose the edge $uv$ is inserted and creates at least one triangle
$uvw$. As explained above, in this case at least one of $u$ and $v$
knows that the triangle exists, therefore it sends $\texttt{ACCEPT} = 1$
to the two other nodes. Thus all three nodes know that they are part of
a triangle. It follows that on each round all nodes can determine
whether they are in a triangle.
\qedhere

\end{proof}

We can extend the algorithm of Theorem~\ref{th:k3_mdtct_ub_pe} to handle
node/edge deletions as well. Now, in addition to \texttt{NEWID}, every
node sends to all its neighbors an indication whether or not it has lost
a neighbor on the current round, along with the ID of the lost neighbor
(if any). Also, every node $u$ sends every neighbor $v$ one additional
bit, indicating whether $u$ received $\texttt{NEWID} = v$ from any other
node on the last round. We denote this bit by \texttt{LAST}.

Now, suppose that on round $t$ the edge $uv$ is inserted and creates a
triangle $uvw$. As in Theorem~\ref{th:k3_mdtct_ub_pe}, after one round
of communication all three nodes can determine that they are part of a
triangle. Moreover, both $u$ and $v$ send the appropriate value of
\texttt{NEWID} to all their neighbors. Then, on round $t+1$, every
common neighbor of $u$ and $v$ sends to both of them
$\texttt{LAST} = 1$. From this information, $u$ and $v$ can deduce the
list of all of their common neighbors.

It follows that on every round, every node knows the exact set of
triangles that contained it on the previous round. This allows every
node to handle node/edge deletions appropriately, and to determine when
all triangles that contained it are gone.

\begin{corollary}\label{cor:k3_mdtct_ub_pe_mvme}
	The deterministic $1$-round bandwidth complexity of \MDTCT{K_3}
	under node deletions and edge deletions/insertions is
	$O(\log n)$.
\end{corollary}

If we only consider sequences of graphs with bounded degree $\Delta$,
and only allow edge insertions, the complexity of
Theorem~\ref{th:k3_mdtct_ub_pe} can be improved to
$O(\sqrt{\Delta \log n})$, using an algorithm similar to that of
Theorem~\ref{th:k3_mlist_ub_pe}, but with $O(\sqrt{\Delta \log n})$ bits
instead of $O(\sqrt{n})$, as follows. On every round, every node sends
to all its neighbors an indication whether or not it has a new neighbor.
We denote this information by \texttt{NEW}. Additionally, whenever a new
edge $uv$ is inserted, the following happens:
\begin{itemize}
	\item
		$u$ sends to $v$ an indication whether or not it knows
		about a triangle that contains $v$.
	\item
		For every round $i$ within the last
		$\ceil{\sqrt{\Delta \log n}}$ rounds, $u$ sends to $v$
		an indication whether or not it has had a new neighbor
		on round $i$.
	\item
		For every round $i$ within the last
		$\ceil{\sqrt{\Delta \log n}}$ rounds, $u$ sends to $v$
		an indication whether or not any of its neighbors has
		sent $\texttt{NEW} = 1$ on round $i$.
	\item
		$u$ computes the list of IDs of all its current
		neighbors, and starts sending it to $v$,
		$\ceil{\sqrt{\Delta \log n}}$ bits on every round. Since
		$u$ has at most $\Delta$ neighbors, after
		$O(\sqrt{\Delta \log n})$ rounds $v$ has the complete
		list. Note that by the time this process completes, the
		list may not be up-to-date.
\end{itemize}

All of this requires a $O(\sqrt{\Delta \log n})$ bandwidth, and it can
be shown, similarly to Theorem~\ref{th:k3_mlist_ub_pe}, that this allows
every node to give the correct output in all cases.

For a quiet round we have an upper bound of $O(1)$ bits from
Observation~\ref{obs:probrel}, Theorem~\ref{th:k3_mlist_ub_r2_pe}, and
Corollary~\ref{cor:k3_mlist_ub_r2_pe_mvme}.

Finally, for node insertions, we have a trivial upper bound of $O(n)$
bits. For a quiet round, the next theorem shows an upper bound of $O(1)$
bits.

\begin{theorem}\label{th:k3_mdtct_ub_r2_pv}
	The deterministic $2$-round bandwidth complexity of \MDTCT{K_3}
	under node insertions is $O(1)$.
\end{theorem}
\begin{proof}

On every round, every node $u$ sends to every neighbor $v$ two bits:
\begin{itemize}
\item
	One bit, denoted by \texttt{NEW}, indicates whether or not $u$
	has a new neighbor which has been inserted on the current round.
\item
	One bit, denoted by \texttt{LAST}, indicates whether or not any
	neighbor of $u$ other than $v$ has sent $u$ $\texttt{NEW} = 1$
	on the last round.
\end{itemize}

Consider a triangle $uvw$ and assume that $w$ is the last node inserted.
When $w$ is inserted, $u$ and $v$ send $\texttt{NEW} = 1$ to each other.
Since only one node can be inserted on a single round, both $u$ and $v$
know that $w$ is a common neighbor, hence the triangle $uvw$ exists.

On the next round, both $u$ and $v$ send $\texttt{LAST} = 1$ to $w$.
This implies that some neighbor of $u$ (resp. $v$) has sent
$\texttt{NEW} = 1$ on the last round, i.e., some neighbor of $u$ (resp.
$v$) has had a new neighbor inserted on the last round. Since only one
node can be inserted on a single round, $w$ can determine that the new
neighbor must be $w$ itself, that is, some neighbor of $u$ (resp. $v$)
is now also a neighbor of $w$. Hence $w$ can determine that it is part
of at least one triangle, and output $1$ as required.
\qedhere

\end{proof}

The above algorithm can be executed in parallel with the algorithm of
Corollary~\ref{cor:k3_mlist_ub_r2_pe_mvme} to handle all four types of
changes. However, we now have a problem combining node insertions and
edge insertions, just as we had a problem in
Corollary~\ref{cor:k3_mlist_ub_r2_pe_mvme} when we combined edge
deletions and node deletions. The problem is that when a node $v$ has a
new neighbor $w$, it does not know whether the edge $vw$ was inserted on
the current round, or $w$ is a new node which was inserted on the
current round and may be connected to other nodes as well. For example,
assume $u$ has two neighbors, $v$ and $w$, and the edge $vw$ does not
exist. Then, either the edge $vw$ is inserted, or a new node $x$ is
inserted and connected to both $v$ and $w$. In both cases, $v$ and $w$
send $\texttt{NEW} = 1$ to $u$, and $u$ cannot distinguish between the
former case (in which it should output 1), and the latter case (in which
it should output 0).

To solve this problem, we can have every new node indicate to all its
neighbors that it is new. Thus, whenever a node $u$ has a new neighbor
$x$, after one round of communication $u$ can determine whether $x$ is a
new node or not and send this information to all its neighbors. This
allows all nodes to distinguish between the two problematic cases we
described, and give the correct output in all cases.

\begin{corollary}\label{cor:k3_mdtct_ub_r2_pv_mvmepe}
	The deterministic $2$-round bandwidth complexity of \MDTCT{K_3}
	under node/edge deletions/insertions is $O(1)$.
\end{corollary}

\subsubsection{Lower bounds}

For node deletions and edge deletions, we have trivial constant lower
bounds, which are tight, as shown above. For edge insertions, we have
shown a general algorithm that uses $O(\log n)$ bits, and also an
algorithm that uses $O(\sqrt{\Delta \log n})$ bits for graphs with
bounded degree $\Delta$. The latter algorithm implies that showing a
lower bound of $B$ bits for the bandwidth complexity of this problem
would require looking at sequences of graphs with degree at least
$\Omega(\frac{B^2}{\log n})$. In particular, in order to show that the
algorithm of Theorem~\ref{th:k3_mdtct_ub_pe} is optimal, one has to
consider sequences of graphs with degree at least $\log n$. Other than
that, it remains an open question whether or not the algorithm of
Theorem~\ref{th:k3_mdtct_ub_pe} can be improved.

\begin{openq}
	What is the $1$-round bandwidth complexity of \MDTCT{K_3} under
	edge insertions?
\end{openq}

Finally, for node insertions, the following theorem shows a lower bound
of $\Omega(n)$ bits.

\begin{theorem}\label{th:k3_mdtct_lb_pv}
	The randomized $1$-round bandwidth complexity of \MDTCT{K_3}
	under node insertions is $\Omega(n)$.
\end{theorem}
\begin{proof}

Let $A$ be a (randomized) $1$-round algorithm that solves \MDTCT{K_3}
under node insertions using bandwidth $B$ with error probability
$\varepsilon$. Fix $x \in V$, let $U = V \setminus \{x\}$, and let
$\mathcal{C}$ be the set of all possible graphs on the nodes of $U$. For
every $C \in \mathcal{C}$ and every $u,v \in U$, define the sequence
$S_{C,u,v}$ as follows:
\begin{itemize}
\item
	Start with an empty graph (no nodes and no edges).
\item
	During $n-1$ rounds, insert one node of $U$ on each round and
	connect it to the nodes which have been already inserted,
	according to the edges in $C$. After $n-1$ rounds, the graph is
	identical to the graph $C$.
\item
	On round $n$ insert $x$ and connect it to $u$ and $v$.
\end{itemize}

Then, $x$ should output 1 iff the edge $uv$ exists in $C$. By our
assumption, for every $C$ and every $u,v \in U$, the output is correct
with probability at least $1-\varepsilon$. Note that during the final
round, $x$ receives only one message from each of $u$ and $v$. Also,
during this final round, $u$ and $v$ do not know the identity of the
other neighbor of $x$. Hence, the message received from $u$ depends only
on the identity of $u$ and the graph $C$, and not on the identity of
$v$, and the message received from $v$ depends only on the identity of
$v$ and the graph $C$.

Now, consider the following experiment for a given graph $C \in
\mathcal{C}$. First, we run the sequence $S_{C,u,v}$ for some $u,v \in
U$, and stop just before the last round, in which $x$ is connected to
$u$ and $v$. Then, every node $w \in U$ generates a message to be sent
to $x$, as if it has been connected to $x$ on the final round. For every
$w \in U$, let $m_w$ be the message generated by $w$ (note that $m_w$ is
a random variable).

Note again, that the messages generated by the nodes of $U$ depend only
on the graph $C$, and not on the other nodes that may have been
connected to $x$ on the last round. Therefore, for every $u,v \in U$,
given the messages generated by $u$ and $v$, $x$ should be able to
determine, with probability at least $1-\varepsilon$, whether or not the
edge $uv$ exists in $C$.

Next, for every pair of nodes $u,v \in U$, let $I_{uv}$ be the output of
$x$ given the two messages $m_u$ and $m_v$. Note that for nodes $u,v,w$,
the variables $I_{uv}$ and $I_{uw}$ are not necessarily independent.
Now, let $C'$ be the graph on the nodes of $U$, in which the edge $uv$
exists if and only if $I_{uv}=1$. We consider $C'$ to be the result of
the experiment.

Let $p_C$ denote the probability that at the end of the experiment we
have $C'=C$. Since, for every $u,v \in U$, the value of $I_{uv}$
corresponds to the edge $uv$ in the graph $C$ with probability at least
$1-\varepsilon$, we have
$
	p_C \geq (1-\varepsilon)^{\binom{n-1}{2}}
$. Summing the above for every $C \in \mathcal{C}$ we get:
\begin{align}\label{eq:mdtct_pv_geq}
		\sum_C p_C
	\geq
		|\mathcal{C}| \cdot (1-\varepsilon)^{\binom{n-1}{2}}
	=
		(2-2\varepsilon)^{\binom{n-1}{2}}
\end{align}

On the other hand, consider the set of messages generated by the nodes
of $U$ at the end of the first stage of the experiment. For every
possible set of messages $M$ generated by the nodes of $U$ during the
first stage of the experiment, denote by $\phi_C(M)$ the probability for
generating exactly the messages of $M$. Also, for every $C' \in
\mathcal{C}$, let $\Psi_M(C')$ denote the probability for the result of
the experiment to be equal to $C'$, given the set of generated messages
$M$. We have:
\A{
		\sum_C p_C
	=
		\sum_C \sum_M \phi_C(M) \cdot \Psi_M(C)
	\leq
		\sum_C \sum_M \Psi_M(C)
	=
		\sum_M \sum_C \Psi_M(C)
	=
		\sum_M 1
	=
		|\mathcal{M}|
}

where $\mathcal{M}$ is the set of all possible values of $M$. Since
every message has exactly $2^B$ possible values, the number of possible
sets of $(n-1)$ messages is
$
	|\mathcal{M}| = 2^{B(n-1)}
$, and hence
$
	\sum_C p_C \leq 2^{B(n-1)}
$. Combining this with~\eqref{eq:mdtct_pv_geq} gives
$
		2^{B(n-1)}
	\geq
		(2-2\varepsilon)^{\binom{n-1}{2}}
$. After some simplifications we get the desired bound:
\[
		B
	\geq
		\log (2-2\varepsilon) \cdot \frac{n-2}{2}
	=
		\Omega(n)
\qedhere
\]
\end{proof}

\subsection{Listing and detection}

Table~\ref{tbl:k3_listdtct} summarizes the results of this subsection.

\begin{table}[tbh]
\centering
\begin{tabular}{|c|c|c|c|}
\hline

	  Node deletions
	& Edge deletions
	& Edge insertions
	& Node insertions
	\\ \hline

	  $0$
	& $\Theta(1)$
	& $\Theta(1)$
	& $\Theta(1)$
	\\ \hline

\end{tabular}
\caption{
	Bandwidth complexities of \LIST{K_3} and \DTCT{K_3}
}
\label{tbl:k3_listdtct}
\end{table}
\subsubsection{Upper bounds}

The upper bounds for node deletions and edge deletions follow from
Observation~\ref{obs:probrel} and the results of the last subsection.

For edge insertions, the algorithm of
Theorem~\ref{th:k3_mlist_ub_r2_pe}, which solves \MLIST{K_3} in $2$
rounds, guarantees that on every round, every triangle is known to at
least one of its nodes. Hence, it can be used to solve \LIST{K_3} in $1$
round.

\begin{corollary}\label{cor:k3_list_ub_pe}
	The deterministic $1$-round bandwidth complexity of \LIST{K_3}
	under edge insertions is $O(1)$.
\end{corollary}

The same is true for the algorithm of
Corollary~\ref{cor:k3_mlist_ub_r2_pe_mvme}, where we also handle
node/edge deletions. Therefore:

\begin{corollary}\label{cor:k3_list_ub_pe_mvme}
	The deterministic $1$-round bandwidth complexity of \LIST{K_3}
	under node deletions and edge deletions/insertions is $O(1)$.
\end{corollary}

Considering node insertions, the problem can be solved with $O(1)$ bits
of bandwidth.

\begin{theorem}\label{th:k3_list_ub_pv}
	The deterministic $1$-round bandwidth complexity of \LIST{K_3}
	under node insertions is $O(1)$.
\end{theorem}
\begin{proof}

The algorithm is as follows: on every round, every node sends to each
one of its neighbors an indication whether or not it got a new neighbor
on the current round. Thus, whenever a node $w$ is inserted and creates
a triangle $uvw$, $u$ receives this indication from $v$ (and vice
versa). Since $u$ also got connected to $w$ on the current round, and
since only one node can be inserted on every round, $u$ can deduce that
the triangle $uvw$ exists.

Note that $v$ also receives the same indication from $u$, and can deduce
that the triangle $uvw$ exists. Thus, this algorithm guarantees that
every triangle is known to at least two of its nodes.
\qedhere

\end{proof}

In order to handle node deletions and edge deletions, we can combine the
above algorithm with the technique used in
Corollary~\ref{cor:k3_mlist_ub_r2_pe_mvme} to distinguish between the
deletion of an edge and the deletion of a node. This technique ensures
that when either a noe or an edge is delete on round $t$, all incident
nodes can determine, on round $t+1$, whether an edge or a node was
deleted. This, combined with the guarantee that every triangle is known
to at least two of its nodes, ensures that on every round at least one
node can determine whether the triangle still exists, even if a node or
an edge was deleted on the current round. Thus, listing all triangles is
possible on every round, as required.

\begin{corollary}\label{cor:k3_list_ub_pv_mvmepe}
	The deterministic $1$-round bandwidth complexity of \LIST{K_3}
	under node/edge deletions/insertions is $O(1)$.
\end{corollary}

By Observation~\ref{obs:probrel}, all upper bounds shown here for
\LIST{K_3} also hold for \DTCT{K_3}.

\subsubsection{Lower bounds}

It is not clear whether it is possible to handle both node insertions
and edge insertions with $O(1)$ bits of bandwidth. The algorithm of
Corollary~\ref{cor:k3_mdtct_ub_pe_mvme} can be used to handle all four
types of changes with $O(\log n)$ bits. Other than that, the exact
bandwidth complexity of this combination is an open question.

\begin{openq}
	What is the $1$-round bandwidth complexity of \LIST{K_3} under
	node/edge insertions?
\end{openq}

\section{Larger cliques}

Throughout this section, let $s \geq 4$ be some constant, and denote by
$K_s$ a clique on $s$ nodes.

\subsection{Membership listing}

Table~\ref{tbl:ks_mlist} summarizes the results of this subsection.

\begin{table}[tbh]
\centering
\begin{tabular}{|l|c|c|c|c|}
\hline

	& Node deletions
	& Edge deletions
	& Edge insertions
	& Node insertions
	\\ \hline
$r = 1$
	& $0$
	& $\Theta(1)$
	& $\Theta(\sqrt{n})$
	& $\Theta(n)$
	\\ \hline
$r \geq 2$
	& $0$
	& $\Theta(1)$
	& $\Theta(1)$
	& $\Theta(n/r)$
	\\ \hline

\end{tabular}
\caption{
	Bandwidth complexities of \MLIST{K_s}
}
\label{tbl:ks_mlist}
\end{table}

\subsubsection{Upper bounds}

For every node $u$, the set of all cliques that contain $u$ is
completely determined by the set of all triangles that contain $u$.
Therefore, every algorithm that solves \MLIST{K_3} can be used to solve
\MLIST{K_s} under the same set of changes. All upper bounds follow thus
from the upper bounds for \MLIST{K_3} given in
Section~\ref{subsubsec:tri-mlist-upper}.

\subsubsection{Lower bounds}

For node deletions and edge deletions we have the trivial lower bounds
of $0$ and $\Omega(1)$, respectively, both of which are tight.

The lower bound for edge insertions is shown next.

\begin{theorem}\label{th:ks_mlist_lb_pe}
	The randomized $1$-round bandwidth complexity of \MLIST{K_s}
	under edge insertions is $\Omega(\sqrt{n})$.
\end{theorem}
\begin{proof}

The proof is almost identical to the proof of
Theorem~\ref{th:k3_mlist_lb_pe}, with only one difference: instead of
starting with an empty tripartite graph, we add a clique on $s-3$ nodes,
which are connected to every node in the graph except $v$. The initial
graph now looks like this:
\begin{figure}[t]
	\begin{center}
	\begin{tikzpicture}[scale=0.2]
	\tikzstyle{every node}+=[inner sep=0pt]

	\draw [black] (20,-30) ellipse (6 and 14);
	\draw (20,-10) node {$W$};
	\draw (20,-20) node {$w_1$};
	\draw (20,-25) node {$\vdots$};
	\draw (20,-30) node {$\vdots$};
	\draw (20,-35) node {$\vdots$};
	\draw (20,-40) node {$w_{n-t-s+2}$};

	\draw [black] (40,-30) ellipse (3 and 8);
	\draw (40,-10) node {$U$};
	\draw (40,-25) node {$u_1$};
	\draw (40,-30) node {$\vdots$};
	\draw (40,-35) node {$u_t$};

	\draw (60,-30) node {$v$};
	
	\draw [black] (36,-47) circle (4);
	\draw (36,-47) node {$K$};

	\draw [black] (59,-30) -- (42.3,-25);
	\draw [black] (59,-30) -- (43,-30);
	\draw [black] (59,-30) -- (42.3,-35);
	\draw [black,dashed,line width=2.8,dash pattern=on 6pt off 6pt] (37,-30) -- (26,-30);
	\draw [black] (60,-31) .. controls (60,-64) and (0,-64) .. (14,-30);
	
	\draw [black,line width=2.8] (39.10,-44.45) -- (59.28,-30.51);
	\draw [black,line width=2.8] (36.92,-43.09) -- (38.4,-36.8);
	\draw [black,line width=2.8] (33.28,-44.11) -- (25.44,-35.78);

	\end{tikzpicture}
	\caption{
		The lower bound sequence for \MLIST{K_3} with edge
		insertions.
	}
	\label{fig:ks}
	\end{center}
\end{figure}
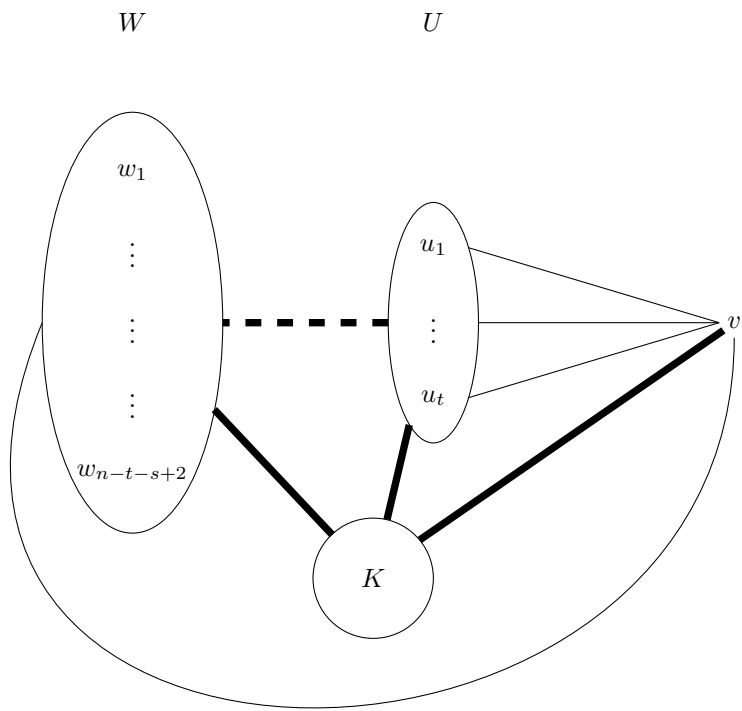

with edges between every node in $K$ and every other node except $v$.

As in Theorem~\ref{th:k3_mlist_lb_pe}, for every bipartite graph $C$
between the nodes of $U$ and $W$ and every node $w \in W$, we define a
sequence of changes $S_{C,w}$ as follows. We start with the initial
graph defined above. Then, we insert edges between $U$ and $W$ to get
the bipartite graph $C$. Then, during the next $t+s$ rounds, we connect
$v$ to every node in $U \cup K$, one by one. Finally, we insert the edge
$vw$. Note that now, since $v$ is connected to all nodes in $U$ and $K$,
the set of instances of $K_s$ that contain $v$ is uniquely determined by
the set of neighbors of $w$ in $U$. Therefore, after a single round of
communication, $v$ should be able to determine the set of all neighbors
of $w$ in $U$.

The proof proceeds as in Theorem~\ref{th:k3_mlist_lb_pe}. The size of
$W$ is now $n-t-s+2$, and the number of messages that $v$ receives from
the nodes of $U$ and $K$ is now $\frac{(t+s)(t+s+3)}{2}$. We end up with
the following inequality, where $|W_1| \geq \alpha |W|$ for some
constants $\alpha,\beta,\gamma \in (0,1)$:
\A{
			\beta
		\cdot
			\gamma^{n-t-s+2}
		\cdot
			2^{
				t(n-t-s+2)
				-
				B \cdot \frac{(t+s)(t+s+3)}{2}
			}
	&\leq
		2^{(B-t) \cdot |W_1| + t(n-t-s+2)}
}

Setting $t = \ceil{\sqrt{n}}$ gives, as in
Theorem~\ref{th:k3_mlist_lb_pe}, $B \geq \Omega(\sqrt{n})$.
\qedhere

\end{proof}

For $r \geq 2$, since edge insertions require communication, we have a
trivial lower bound of $\Omega(1)$ bits.

For node insertions, we again use a similar construction to the one we
used for triangles to show a lower bound of $\Omega(n/r)$ bits.

\begin{theorem}\label{th:ks_mlist_lb_r_pv}
	For every $r$, the randomized $r$-round bandwidth complexity of
	\MLIST{K_s} under node insertions is $\Omega\pa{\frac{n}{r}}$.
\end{theorem}
\begin{proof}

The proof is almost identical to Theorem~\ref{th:k3_mlist_lb_r_pv}, with
the only difference being that we limit the possible graphs that are
constructed on the first $n-1$ rounds to a certain family of graphs.

Fix $x \in V$, and fix a set of $s-3$ other nodes $K \subseteq
V \setminus \{x\}$. Let $U = V \setminus (K \cup \{x\})$ be the set of
all other nodes, and let $U = L \cup R$ be a partition of $U$ into two
subsets of equal size (if $|U|$ is odd, let $|L| = |R| + 1$). Let
$\mathcal{C}$ be the set of all possible bipartite graphs on the sets
$L$ and $R$. For every $C \in \mathcal{C}$ we define a sequence of
changes $S_C$ as follows:
\begin{itemize}
\item
	Start with an empty graph (no nodes and no edges).
\item
	During $n-(s-2)$ rounds, insert one node of $U$ on each round
	and connect it to the nodes which have been already inserted,
	according to the edges in $C$. After $n-(s-2)$ rounds, the graph
	is identical to the graph $C$.
\item
	During the next $s-3$ rounds, insert one node of $K$ on each
	round and connect it to all other nodes in the graph. After this
	stage, the graph consists of the graph $C$, a clique $K$, and
	all possible edges between the nodes of $C$ and the nodes of
	$K$.
\item
	On round $n$ insert $x$ and connect it to all other nodes in the
	graph.
\end{itemize}

After one round of communication, $x$ should be able to determine the
set of all isntances of $K_s$ that contain it. Note that, for $C,C' \in
\mathcal{C}$, if an edge $uv$ exists in $C$ but not in $C'$ then the set
$\{x,u,v\} \cup K$ is a clique in the final graph of $S_C$, but not in
the final graph of $S_{C'}$. Therefore, the correct output of $x$ is
different for every two distinct $C,C' \in \mathcal{C}$.

The proof proceeds as in Theorem~\ref{th:k3_mlist_lb_r_pv}. The number
of possible sequences is now
$
		2^{|L| \cdot |R|}
	=
		2^{
			\floor{\frac{n-s+2}{2}}
			\cdot
			\ceil{\frac{n-s+2}{2}}
		}
$, and we end up with the following inequality:
\A{
		2^{r \cdot B \cdot (n-1)}
	\geq
		(1-\varepsilon)
		\cdot
		2^{
			\floor{\frac{n-s+2}{2}}
			\cdot
			\ceil{\frac{n-s+2}{2}}
		}
}

After some simplifications we obtain $r \cdot B = \Omega(n)$, as
claimed.
\qedhere

\end{proof}

\subsection{Membership detection}

Table~\ref{tbl:ks_mdtct} summarizes the results of this subsection.

\begin{table}[tbh]
\centering
\begin{tabular}{|l|c|c|c|c|}
\hline

	& Node deletions
	& Edge deletions
	& Edge insertions
	& Node insertions
	\\ \hline
$r = 1$
	& $0$
	& $\Theta(1)$
	& $O(\sqrt{n})$
	& $\Theta(n)$
	\\ \hline
$r \geq 2$
	& $0$
	& $\Theta(1)$
	& $\Theta(1)$
	& $\Theta(1)$
	\\ \hline

\end{tabular}
\caption{
	Bandwidth complexities of \MDTCT{K_s}
}
\label{tbl:ks_mdtct}
\end{table}

\subsubsection{Upper bounds}

All upper bounds follow from Observation~\ref{obs:probrel} and the
results of the last subsection, except for the upper bound of $O(1)$ in
$2$ rounds under node insertions, which can be solved by the same
algorithm used in Theorem~\ref{th:k3_mdtct_ub_r2_pv} and its corollary
for triangles.

\subsubsection{Lower bounds}

The lower bounds for node deletions and edge deletions are trivial.

As for edge insertions, it is not clear whether the upper bound of
$O(\sqrt{n})$ is tight, and in general, what is the exact bandwidth
complexity of the problem.

Note that the algorithm of Theorem~\ref{th:k3_mdtct_ub_pe}, which solves
\MDTCT{K_3} with $O(\log n)$ bits of bandwidth, does not work for larger
cliques. To understand why, consider an almost-clique on $s$ nodes, with
one missing edge $uv$. That is, inserting the edge $uv$ would create a
clique on $s$ nodes.

For triangles, this is simply a node which is connected to both $u$ and
$v$. In this case, we can guarantee that at least one of the nodes $u$
and $v$ knows that this construct exists. Thus, when the edge $uv$ is
inserted, either $u$ or $v$ can indicate to the two other nodes that
they are part of a triangle. For larger cliques, we have more than one
common neighbor and more than two edges, and it is not clear how to
provide the same guarantee. Therefore, the bandwidth complexity of this
problem remains an open question.

\begin{openq}
	What is the $1$-round bandwidth complexity of \MDTCT{K_s} under
	edge insertions?
\end{openq}

For node insertions, we show that at least $\Omega(n)$ bits of bandwidth
are required.

\begin{theorem}\label{th:ks_mdtct_lb_pv}
	The randomized $1$-round bandwidth complexity of \MLIST{K_s}
	under node insertions is $\Omega(n)$.
\end{theorem}
\begin{proof}

Our construction is a generalization of the construction used in
Theorem~\ref{th:k3_mdtct_lb_pv} for triangles.

Let $A$ be a (randomized) $1$-round algorithm that solves \MDTCT{K_3}
under node insertions using bandwidth $B$ with error probability
$\varepsilon$. Fix $x \in V$, and fix a set of $s-3$ other nodes $K
\subseteq V \setminus \{x\}$. Let $U = V \setminus (K \cup \{x\})$, and
let $U = L \cup R$ be a partition of $U$ into two subsets of equal size
(if $|U|$ is odd, let $|L| = |R| + 1$). Let $\mathcal{C}$ be the set of
all possible bipartite graphs on the setst $L$ and $R$. For every $C \in
\mathcal{C}$ and every $u \in L, v \in R$, define the sequence
$S_{C,u,v}$ as follows:
\begin{itemize}
\item
	Start with an empty graph (no nodes and no edges).
\item
	During $n-(s-2)$ rounds, insert one node of $U$ on each round
	and connect it to the nodes which have been already inserted,
	according to the edges in $C$. After $n-(s-2)$ rounds, the graph
	is identical to the graph $C$.
\item
	During the next $s-3$ rounds, insert one node of $K$ on each
	round and connect it to all other nodes in the graph. After this
	stage, the graph consists of the graph $C$, a clique $K$, and
	all possible edges between the nodes of $C$ and the nodes of
	$K$.
\item
	On round $n$ insert $x$ and connect it to $u$, $v$, and all
	nodes of $K$.
\end{itemize}

In the final graph, $x$ is connected to exactly $s-1$ nodes, all of
which are connected to each other except maybe $u$ and $v$. Thus, $x$ is
part of a clique of size $s$ if and only if the edge $uv$ exists.
Therefore, after one round, $x$ should output 1 if the edge $uv$ exists
in $C$, or 0 otherwise. By our assumption, for every $C$ and every $u,v
\in U$, the output is correct with probability at least $1-\varepsilon$.

During the final round, $x$ receives only one message from $u$, one
message from $v$, and $s-3$ messages from the nodes of $K$. As in
Theorem~\ref{th:k3_mdtct_lb_pv}, when $x$ receives these messages from
$u$, $v$, and the nodes of $K$, these nodes do not know the identity of
$u$ and $v$. Therefore, the message received from $u$ depends only on
the identity of $u$ and the graph $C$, and not on the identity of $v$.
Likewise, the message received from $v$ depends only on the identity of
$v$ and the graph $C$, and the $s-3$ messages received from the nodes of
$K$ depend only on the graph $C$.

Now, consider the following experiment for a given graph $C \in
\mathcal{C}$. First, we run the sequence $S_{C,u,v}$ for some $u,v \in
U$, and stop just before the last round, in which $x$ is connected to
$u$ and $v$. Then, every node $w \in U \cup K$ generates a message to be
sent to $x$, as if it has been connected to $x$ on the final round. For
every $w \in U$, let $m_w$ be the message generated by $w$, and let
$m_K$ be the set of messages generated by the nodes of $K$ (note that
$m_w$ and $m_K$ are random variables).

Note again, that the messages generated by the nodes of $U$ depend only
on the graph $C$, and not on the other nodes that may have been
connected to $x$ on the last round. Therefore, for every $u \in L, v \in
R$, given the messages generated by $u$, $v$, and all nodes of $K$, $x$
should be able to determine, with probability at least $1-\varepsilon$,
whether or not the edge $uv$ exists in $C$.

Next, for every pair of nodes $u \in L, v \in R$, let $I_{uv}$ be the
output of $x$ given the messages $m_u$, $m_v$, and $m_K$. Now, let $C'$
be the graph on the nodes of $U$, in which the edge $uv$ exists if and
only if $I_{uv}=1$. We consider $C'$ to be the result of the experiment.

We now compute the probability that the $x$ has the correct output for
all pairs $u \in L, v \in R$. For $C \in \mathcal{C}$, let $p_C$ denote
the probability that at the end of the experiment we have $C'=C$. For
every $u \in L, v \in R$, the value of $I_{uv}$ corresponds to the edge
$uv$ in the graph $C$ with probability at least $1-\varepsilon$. Thus we
have:
\A{
		p_C
	\geq
		(1-\varepsilon)^{|L| \cdot |R|}
}

Summing the above for every $C \in \mathcal{C}$ we get:
\begin{align}\label{eq:ks_mdtct_pv_geq}
		\sum_C p_C
	\geq
		|\mathcal{C}| \cdot (1-\varepsilon)^{|L| \cdot |R|}
	=
		(2-2\varepsilon)^{|L| \cdot |R|}
\end{align}

On the other hand, consider the set of messages generated by the nodes
of $U \cup K$ at the end of the first stage of the experiment. For every
possible set of messages $M$ generated by the nodes of $U \cup K$ during
the first stage of the experiment, denote by $\phi_C(M)$ the probability
for generating exactly the messages of $M$. Also, for every $C' \in
\mathcal{C}$, let $\Psi_M(C')$ denote the probability for the result of
the experiment to be equal to $C'$, given the set of generated messages
$M$. We have:
\A{
		\sum_C p_C
	&=
		\sum_C \sum_M \phi_C(M) \cdot \Psi_M(C)
	\leq
		\sum_C \sum_M \Psi_M(C)
	=
		\sum_M \sum_C \Psi_M(C)
	\\ &=
		\sum_M 1
	=
		|\mathcal{M}|
}

where $\mathcal{M}$ is the set of all possible values of $M$. Since
every message has exactly $2^B$ possible values, the number of possible
sets of $(n-1)$ messages is:
\A{
	|\mathcal{M}| = 2^{B(n-1)}
}

Hence:
\A{
	\sum_C p_C \leq 2^{B(n-1)}
}

Combining this with~\eqref{eq:mdtct_pv_geq} gives:
\A{
	2^{B(n-1)} \geq (2-2\varepsilon)^{|L| \cdot |R|}
}

Both $L$ and $R$ contain $\Theta(n)$ nodes. Therefore, after some
simplifications we get the desired bound:
\[
		B
	\geq
			\log (2-2\varepsilon)
		\cdot
			\frac{|L| \cdot |R|}{n-1}
	=
		\Omega(n)
\qedhere
\]

\end{proof}

\subsection{Listing and detection}

Table~\ref{tbl:ks_listdtct} summarizes the results of this subsection.

\begin{table}[tbh]
\centering
\begin{tabular}{|c|c|c|c|}
\hline

	  Node deletions
	& Edge deletions
	& Edge insertions
	& Node insertions
	\\ \hline

	  $0$
	& $\Theta(1)$
	& $\Theta(1)$
	& $\Theta(1)$
	\\ \hline

\end{tabular}
\caption{
	Bandwidth complexities of \LIST{K_s} and \DTCT{K_s}
}
\label{tbl:ks_listdtct}
\end{table}

\subsubsection{Upper bounds}

The upper bounds for node deletions and edge deletions follow from
Observation~\ref{obs:probrel} and the results of the last subsection.

For edge insertions, the algorithm used in
Corollary~\ref{cor:k3_list_ub_pe} for triangle listing does not work for
larger cliques. However, we can use the algorithm of
Theorem~\ref{th:k3_mlist_ub_r2_pe}, which solves membership-listing of
triangles in $2$ rounds, to solve listing of larger cliques in $1$
round.

\begin{theorem}\label{th:ks_list_ub_pe}
	The deterministic $1$-round bandwidth complexity of \MDTCT{K_s}
	under edge insertions is $O(1)$.
\end{theorem}
\begin{proof}

As said, we use the exact same algorithm of
Theorem~\ref{th:k3_mlist_ub_r2_pe}. This algorithm allows every node $u$
to determine, on round $t$, the list of all triangles that contained $u$
on round $t-1$. Hence, each node can determine the set of all cliques
--- of any size --- that contained it on the previous round.

We show that every clique of size $s$ is known to at least one node. For
this, we recall that by the algorithm of
Theorem~\ref{th:k3_mlist_ub_r2_pe}, on every round, every node sends to
all its neighbors an indication whether or not it got a new neighbor on
the current round. This information is denoted by \texttt{NEW}.

Now, assume that the insertion of an edge $uv$ on round $t$ creates a
clique $K$ of size $s$, and let $x$ be a node in $K$ other than $u$ and
$v$. We show that on round $t$, $x$ can determine that $K$ is a clique.
On round $t-1$, all nodes of $K$ are connected by an edge except $u$ and
$v$. In particular, on round $t-1$, $K \setminus \{u\}$ and $K \setminus
\{v\}$ are both cliques that contain $x$. Hence, as stated above, after
round $t$, $x$ knows that these two cliques exist. Moreover, on round
$t$, both $u$ and $v$ send $\texttt{NEW} = 1$ to $x$, from which $x$ can
deduce that the edge $uv$ has just been inserted. Hence, $x$ can
determine on round $t$ that the nodes of $K$ are all connected to each
other, as required.
\qedhere

\end{proof}

Combining this with node/edge deletions can be done similarly with the
algorithm of Corollary~\ref{cor:k3_mlist_ub_r2_pe_mvme}.

\begin{corollary}\label{cor:ks_list_ub_pe_mvme}
	The deterministic $1$-round bandwidth complexity of \MDTCT{K_s}
	under node deletions and edge deletions/insertions is $O(1)$.
\end{corollary}

For node insertions, we can use the same algorithm used for triangles in
Theorem~\ref{th:k3_list_ub_pv}.

\begin{theorem}\label{th:ks_list_ub_pv}
	The deterministic $1$-round bandwidth complexity of \LIST{K_s}
	under node insertions is $O(1)$.
\end{theorem}
\begin{proof}

We use the same algorithm as in Theorem~\ref{th:k3_list_ub_pv}. By this
algorithm, on every round, each node sends to all its neighbors an
indication whether or not it has a new neighbor on the current round. We
show that this guarantees that every clique on $s$ nodes is known to at
least one node.

Consider a clique $K$ with $s$ nodes, and assume that $u$ is the first
node of $K$ that has been inserted (if some of the nodes of $K$ exist in
the initial graph, let $u$ be any one of them). We show that $u$ can
determine that $K$ is a clique $1$ round after the last node of $K$ is
inserted.

By induction, we show that for every non-empty $L \subseteq K$, $u$ can
determine that $L$ is a clique $1$ round after the last node of $L$ is
inserted. The claim then follows for $L=K$.

The claim is trivial for $L=1$, since it contains only $u$. It is also
trivial if $L$ is the set of nodes of $K$ that exist in the initial
graph. Now, assume that a new node $v \in K$ is inserted, along with
edges to all nodes of $L$. When this happens, $u$ is connected to all
nodes of $L$, therefore it receives from each of them an indication that
they are connected to the new node. Thus, after a single round of
communication, $u$ can determine that the set $L \cup \{v\}$ is a
clique, as claimed.
\qedhere

\end{proof}

By Observation~\ref{obs:probrel}, all upper bounds shown here for
\LIST{K_s} also hold for \DTCT{K_s}.

We note that all results of this subsection hold for any value of $s$,
even $s=n-O(1)$.

\subsubsection{Lower bounds}

The lower bounds for each type of change are either $0$ or $\Omega(1)$,
which is trivial.

The algorithm of Theorem~\ref{th:ks_list_ub_pv} also works in
combination with either node deletions or edge deletions, by simply
having each node send an indication whether it lost a neighbor on the
current round. However, unlike the case of triangles, this does not work
if we allow both node deletions and edge deletions. This is because this
algorithm guarantees that every clique is known to at least two of its
nodes. For triangles, this suffices to guarantee that when either a node
or an edge is deleted, at least one of these two nodes know whether the
triangle still exists. For larger cliques, this is not the case --- for
example, consider a clique on $s \geq 4$ nodes, and let $u,v,x,y$ be
four nodes in that clique. Assume that only $u$ and $v$ know that the
clique exists. Now, either the edge between $x$ and $y$ is deleted, or
some other node is deleted which is connected to both $x$ and $y$, but
not to either $u$ or $v$. In both cases $x$ and $y$ can send an
indication that they lost a neighbor, but until the next round, $u$ and
$v$ cannot distinguish between the two cases, and therefore cannot
determine whether the clique still exists.

This can be fixed by having every node send, together with an indication
whether or not it has a new neighbor or a missing neighbor, the ID of
the new or missing neighbor. However, this requires $O(\log n)$ bits of
bandwidth, which is significantly more than the $O(1)$ bits required to
handle each type of change alone. Thus, the exact bandwidth complexity
of this combination remains an open question.

\begin{openq}
	What is the $1$-round bandwidth complexity of \LIST{K_s} under
	node insertions and node/edge deletions?
\end{openq}

Also, as is the case for triangles, it is an open question whether it is
possible to combine edge insertions and node insertions with $O(1)$ bits
of bandwidth.

\begin{openq}
	What is the $1$-round bandwidth complexity of \LIST{K_s} under
	node/edge insertions?
\end{openq}

\paragraph{Acknowledgements.}
This project has received funding from the European Union’s Horizon 2020
Research And Innovation Program under grant agreement no. 755839.

\bibliography{related}

\end{document}